\documentclass[lettersize,journal]{IEEEtran}
\usepackage{amsmath,amsfonts}
\usepackage{amsthm}
\usepackage{hyperref}
\usepackage{array}
\usepackage{caption}
\usepackage[caption=false,font=normalsize,labelfont=sf,textfont=sf]{subfig} 
\usepackage{textcomp}
\usepackage{stfloats}
\usepackage{url}
\usepackage{verbatim}
\usepackage{graphicx}
\usepackage[dvipsnames]{xcolor}
\usepackage{tabularx}
\usepackage{booktabs}
\usepackage{overpic}
\usepackage{stfloats}
\usepackage{enumitem}
\usepackage{algorithm}
\usepackage{xcolor}
\usepackage{amsmath}
\usepackage{cite}
\newtheorem{Proposition}{Proposition}
\newtheorem{remark}{Remark}

\usepackage{algorithm}
\usepackage{algpseudocode} % from algorithmicx
\usepackage{amsmath,amssymb}

\usepackage{threeparttable}

\def\BibTeX{{\rm B\kern-.05em{\sc i\kern-.025em b}\kern-.08em
    T\kern-.1667em\lower.7ex\hbox{E}\kern-.125emX}}
\usepackage{balance}
\begin{document}
% \title{Early Intervention Strategies to Enhance Fairness and Efficiency in Autonomous Traffic Flow Management}

\title{Intervention Strategies for Fairness and Efficiency at Autonomous Single-Intersection Traffic Flows}

\author{Salman Ghori, Ania Adil, Melkior Ornik, Eric Feron
\thanks{Salman Ghori, Ania Adil, and Eric Feron are with Computer, Electrical and Mathematical  Science $\&$ Engineering Division (CEMSE), KAUST, Thuwal 23955-6900, Saudi Arabia (email: {salman.ghori@kaust.edu.sa, ania.adil@kaust.edu.sa, eric.feron@kaust.edu.sa}). Melkior Ornik is with the Department of Aerospace Engineering and the Coordinated Science Laboratory, University of Illinois at Urbana-Champaign, Urbana, IL 61801 USA (email: mornik@illinois.edu).}
\markboth{IEEE TRANSACTIONS ON INTELLIGENT TRANSPORTATION SYSTEMS}{IEEE TRANSACTIONS ON INTELLIGENT TRANSPORTATION SYSTEMS}%
}

\maketitle

\begin{abstract}
%\textcolor{red}{Intersections present significant challenges in traffic management, where maintaining operational requirements and ensuring safety are essential for efficient flow.}
Intersections present significant challenges in traffic management, where ensuring safety and efficiency is essential for effective flow. However, these goals are often achieved at the expense of fairness, which is critical for trustworthiness and long-term sustainability. This paper investigates how the timing of centralized intervention affects the management of autonomous agents at a signal-less, orthogonal intersection, while satisfying safety constraints, evaluating efficiency, and ensuring fairness. A mixed-integer linear programming (MILP) approach is used to optimize agent coordination within a circular control zone centered at the intersection. We introduce the concept of fairness, measured via pairwise reversal counts, and incorporate fairness constraints into the MILP framework. We then study the relationship between fairness and system efficiency and its impact on platoon formation. Finally, simulation studies analyze the effectiveness of early versus late intervention strategies and fairness-aware control, focusing on safe, efficient, and robust management of agents within the control zone.

\end{abstract}


\section*{Nomenclature}

\subsection*{Indices, Sets, and Spaces}
\begin{IEEEdescription}[\IEEEsetlabelwidth{$\mathbb{R}^{n\times m}$}]
  \item[$\mathbb{R}$] Set of real numbers.
  \item[$\mathbb{R}_+$] Nonnegative reals: $\{\,x\in\mathbb{R}:x\ge 0\}$.
  \item[$\mathbb{R}_+^{\ast}$] Strictly positive reals: $\{\,x\in\mathbb{R}:x>0\}$.
  \item[$\mathbb{R}^n$] $n$-dimensional Euclidean space (column vectors).
  \item[$\mathbb{R}^n_+$] Nonnegative orthant: $\{\,x\in\mathbb{R}^n : x_i\ge 0\;\forall\,i\}$.
  \item[$\mathbb{R}^{n\times m}$] Set of real $n\times m$ matrices.
  \item[$G$] Total number of agents.
  \item[$N$] Total number of discrete time steps in the horizon.
  \item[$q,r$] Agent indices, $q,r \in \{\,1,\dots,G\}$.
  \item[$i$] Discrete time index, $i \in \{\,0,\dots,N\}$.
  
\end{IEEEdescription}

\subsection*{Decision Variables}
\begin{IEEEdescription}[\IEEEsetlabelwidth{$s_{q,i,k}\in\mathbb{R}_+$, \quad $k=1,2,3,4$}]
  \item[$p_{q,i}^x,\,p_{q,i}^y\in\mathbb{R}$]  
    Position [m] of agent $q$ in the $x$- and $y$-coordinates at time $i$.

  \item[$v_{q,i}^x,\,v_{q,i}^y\in\mathbb{R}$]  
    Velocity [m/s] of agent $q$ in the $x$- and $y$-directions at time $i$.

  \item[$x_{q,i}\in\mathbb{R}^4$]  
    State vector of agent $q$ at time $i$,
    \[
      x_{q,i} = \bigl(p_{q,i}^x,\;p_{q,i}^y,\;v_{q,i}^x,\;v_{q,i}^y\bigr)^\top.
    \]

  \item[$u_{q,i}\in\mathbb{R}^2$]  
    Control input [m/s$^{2}$] of agent $q$ at time $i$,
    \[
      u_{q,i} = \bigl(u_{q,i}^x,\;u_{q,i}^y\bigr)^\top.
    \]

  \item[$s_{q,i,k}\in\mathbb{R}_+$, \quad $k=1,\dots,4$]  
    Slack variables for the absolute deviation of each state component from its desired final value at time $i$.

  \item[$V_{q,i}\in\mathbb{R}_+$]  
    Velocity [m/s] of agent $q$ at time $i$; used in the objective to reward higher speed.

  \item[$O_{q,r}\in\{0,1\}$]  
    Binary variable indicating priority reversal between agents $q$ and $r$.

  \item[$\varpi_{q,i},\, \varpi_{r,i},\, \mathrm{cl}_{q,r,i}\in\{0,1\}$]  
    Binary variables

  \item[$d_{q,i},\,d_{r,i} \, \delta_{q,r,i} \in\mathbb{R}_+$]  
    Auxiliary variables

\end{IEEEdescription}

\subsection*{Parameters and Constants}
\begin{IEEEdescription}[\IEEEsetlabelwidth{$A_{q}\in\mathbb{R}^{4\times4},\,B_{q}\in\mathbb{R}^{4\times2}$}]

\item[$R \in\mathbb{R}_+$]  
    Radius [m] of the control zone.
    
  \item[$x_{q,0}\in\mathbb{R}^4$]  
    Initial state of agent $q$:
    \[
      x_{q,0} = \bigl(p_{q,0}^x,\;p_{q,0}^y,\;v_{q,0}^x,\;v_{q,0}^y\bigr)^\top.
    \]

  \item[$x_{q,f}\in\mathbb{R}^4$]  
    Desired final state of agent $q$:
    \[
      x_{q,f} = \bigl(p_{q,f}^x,\;p_{q,f}^y,\;v_{q,f}^x,\;v_{q,f}^y\bigr)^\top.
    \]

  \item[$A_{q}\in\mathbb{R}^{4\times4},\,B_{q}\in\mathbb{R}^{4\times2}$]  
    State‐transition and control‐input matrices for agent $q$.
    % (discrete‐time linear dynamics).

  \item[$u_{q,\min},\,u_{q,\max}\in\mathbb{R}^2$]  
    Lower and upper bounds on control inputs for agent $q$

  \item[$x_{q,\min},\,x_{q,\max}\in\mathbb{R}^4$]  
    Lower and upper bounds on the state variables  of agent $q$

  \item[$\alpha,\,\beta\in\mathbb{R}_+^4$]  
    Weight vectors for state deviation.

  \item[$\gamma \in \mathbb{R}_+$]  
    Weight for velocity $V_{q, i}$. 
    % in the objective (higher $\gamma$ encourages higher speed).

  \item[$\lambda \in \mathbb{R}_+$]  
    Weight for priority‐reversal $O_{q,r}$.
    % in the objective (penalizes reversals).

  \item[$c_x,c_y\in\mathbb{R}$]  
    Coordinates of the intersection center. 
    % (used in intersection and crossing constraints).

\item[$d_{\mathrm{safe}} \in \mathbb{R}_{+}^{\ast}$]  
    Minimum safety distance [m] between any two agents within the same flow.

\item[$d_{\mathrm{sep}} \in \mathbb{R}_{+}^{\ast}$]  
    Minimum cross-flow separation distance [m] at the intersection.

  \item[$\Lambda_x$, $\Lambda_y$]  
    Poisson arrival rates [s$^{-1}$] for x--flow, y--flow.

  \item[$b\in\mathbb{R}_+$]  
     Threshold to identify possible reversal of agents $q,r$.

  \item[$W,M \in \mathbb{R}_+$]  
    Sufficiently large constant (“Big‐$M$”).

  \item[$\epsilon \in \mathbb{R}_{+}^{\ast}$]  
    Small positive constant. 

  \item[$d_{\mathrm{stop}}\in\mathbb{R}_{+}^{\ast}$]
    Safe braking distance [m] to stop completely (worst-case).
    
  \item[$\delta_{\mathrm{platoon}}\in\mathbb{R}_+$]
    Length of the longest platoon.
  \item[$n_{\mathrm{platoon}}\in\mathbb{R}_+$]
    Largest trailing agent count in either flow.
  \item[$n_{1},n_{2}\in\mathbb{R}_+$]
    Platoon sizes in x/y--flows.
    
\end{IEEEdescription}

\section{Introduction}

\IEEEPARstart{I}{ntersections} represent constrained spaces where two or more distinct flows converge, presenting significant challenges and opportunities in traffic management~\cite{10669156, dresner2008multiagent, nagrare2024intersection, CAVs3407903, sengupta2025urban, khanmohamadi2025smart}. Effective traffic management is important to maximize capacity and throughput, reduce overall delay, ensure safety, and improve fuel efficiency. Researchers have extensively investigated both light-controlled~\cite{kim2023optimal,7349244,little1966synchronization, 11004127} and signal-less traffic intersections~\cite{10833692, 8732975,dresner2004multiagent, carlino2013auction, tlig2012cooperative}. However, as transportation evolves towards greater autonomy, managing intersections without traffic signals becomes increasingly applicable across diverse contexts, including autonomous vehicle fleets at intersecting streets~\cite{9535372, nisyrios2025optimization}, satellites or spacecraft in different orbital planes that may have the same crossing windows~\cite{thangavel2024simulation, sorge2025space}, airport runway scheduling~\cite{hardell2025optimizing, 928721, 5717044, frazzoli2001resolution, 4298903}, and automated warehouse robot sorting goods~\cite{wurman2008coordinating}.

Despite differences in scale and domain, the above-mentioned scenarios share a common structural challenge in coordinating intersecting flows to prevent collisions while optimizing efficiency. A canonical example involves two orthogonal flows crossing a signal-less intersection, which raises key questions about safety guarantees and performance objectives. In particular, the timing of intervention by a centralized controller to manage the intersection has received little attention in the literature. Related work that has a similar theme to the timing of intervention includes Tlig et al.’s synchronization‐based intersection controller \cite{tlig2014decentralized} and Lim et al.’s Extended Arrival Manager for air traffic~\cite{jun2022towards}, although not an intersection controller, shares the key idea of earlier centralized coordination.

Beyond efficiency, effective traffic management must address fairness, essential for operational trustworthiness~\cite{kleinberg1999fairness}, stakeholder endorsement~\cite{marsh1994equity}, and long-term sustainability~\cite{bertsimas2009fairness}. Ignoring fairness risks systemic discrimination and stakeholder discontent~\cite{kumar2006fairness,bertsimas2011price}, particularly under high traffic density conditions. Although intuitive, fairness usually lacks a universal formal definition~\cite{soomer2008fairness} because perceptions of fairness can vary based on context, perspective, or objectives. Fairness metrics~\cite{chin2020tradeoffs} can be categorized as reversal~\cite{chin2022efficient}, overtaking~\cite{bertsimas2016fairness}, and time-ordered deviation~\cite{barnhart2012equitable}. Importantly, when agents cross the intersection in an order different from their arrival order at the control zone, not strictly following the first-in–first-out (FIFO) rule, they undergo a reversal. This can facilitate platooning, i.e., short-headway groups of agents with near-synchronized velocities~\cite{4840432,1470286}. These dense formations cross as units, improving throughput~\cite{hall2005vehicle}, reducing delay~\cite{9246221}, and lowering energy consumption~\cite{10021253}. 

In our prior work~\cite{11108008}, we studied early versus late intervention, modeled as large versus small control zones, at a fixed traffic density and identified a delay-minimizing optimal zone radius. We empirically showed that performance saturates beyond a threshold radius. We also observed platoon formation driven by reversal-enabled sequencing under collision constraints. While our previous study successfully demonstrated these efficiency gains, it did not evaluate the consequential impacts on system fairness.

Building on these findings, the present work develops a systematic study of fairness across intervention strategies {using MILP}. While early intervention offers greater scheduling flexibility that may improve fairness, late intervention can better respect FIFO order, but often at an efficiency cost. We address this gap by elevating fairness to a core objective and examining how fairness constraints interact with control zone sizing across traffic densities. Within our framework, collision-avoidance constraints naturally encourage platooning. Allowing reversals near the intersection can delay one agent slightly to let another pass first, releasing a tightly packed sequence behind it and compressing flows for efficiency gains. Conversely, prioritization of FIFO suppresses platooning, improving fairness but reducing efficiency.
\noindent This paper significantly expands our framework with three key contributions:
\begin{itemize}
\item{We explicitly integrate fairness constraints into the intersection management framework, demonstrating how it reshapes scheduling decisions and reveals equity–efficiency trade-offs.}
\item {We show how control zone size critically influences reversals, platoon formation, and how these factors impact overall system performance.}
\item {We conduct a comprehensive evaluation of delay and energy consumption as performance metrics across varying traffic densities, leading to the identification of optimal control zone radii that balance efficiency and fairness.}
\end{itemize}
These insights help bridge gaps in intersection management research and inform resilient designs for roads, airspace, runways, and logistics. Throughout, “agent” denotes vehicles, UAVs, spacecraft, or robots in cross-flow scenarios.

This paper is organized as follows. Section~\ref{sec:prob_desc} describes the problem setup, focusing on multi-agent behavior at an intersection. Section~\ref{sec:model} details the MILP formulation. Section~\ref{sec:fairness-measure} defines fairness metrics. Section~\ref{sec_simulation} describes the simulation framework. Section~\ref{results} presents results on optimal radii and proof of existence of an optimal radius, fairness trade-offs, and platooning dynamics. Section~\ref{conclusion} concludes and outlines future work.

\section{Problem Description} \label{sec:prob_desc}
This paper addresses the challenge of coordinating autonomous agents from two orthogonal, single-lane flows at a signal-less intersection. Each agent is represented by a point mass, whose dynamics are linear and move from north to south or west to east. To safely and efficiently manage their interaction, we define a circular control zone of radius $R$, centered at the intersection, as depicted in Fig.~\ref{fig:frame_780_1}. Agents operate autonomously outside this zone, typically maintaining a maximum cruising velocity. Upon entering the control zone, their control is transferred to a centralized authority that optimizes their trajectories until they exit. The radius of the control zone $R$ is not explicitly used in the MILP formulation (see Section~\ref{sec:model}). Rather, it specifies when centralized coordination becomes active to manage agents approaching the intersection.\\

\begin{figure}[!ht]
    \centering
    \includegraphics[width=\columnwidth]{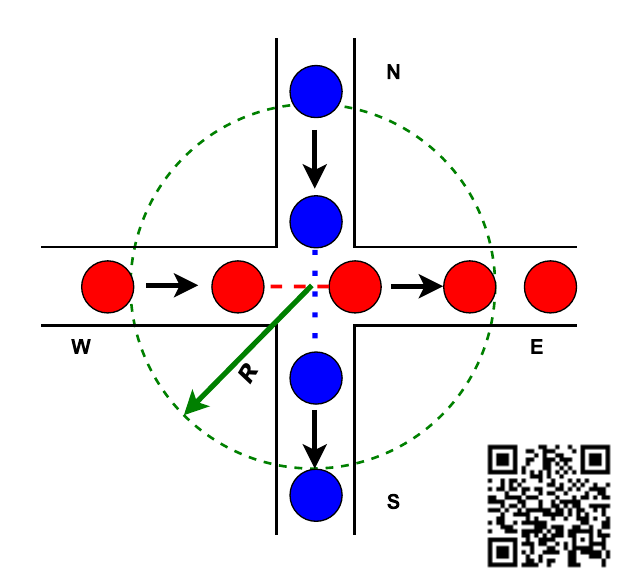}
   \caption{Illustration of the control zone circle (radius \(R\)) in an intersection scenario. The QR code links to a video demonstrating the control zone implementation.}
    \label{fig:frame_780_1}
\end{figure}
\vspace{-10pt} % reduce space before paragraph starts

The controller’s primary objectives are 
\begin{itemize}
\item \textit{Safety:} Maintain a minimum safety distance between agents of the same flow and a minimum cross-flow separation from agents in the orthogonal flow, avoiding collisions at the intersection.

\item \textit{Efficiency:} Minimize total delay and energy consumption by maximizing the total velocities of agents throughout the control zone. 

\item \textit{Fairness:} Maintain FIFO crossing order as much as possible to avoid disproportionate delays and promote equitable treatment.
\end{itemize}

However, achieving these goals together is important due to inherent trade-offs. In particular, the timing of intervention, determined by the control zone radius $R$, introduces a principal design choice. A larger $R$ enables earlier intervention, giving the controller greater planning flexibility to coordinate agents proactively. Such early intervention can preserve fairness and improve overall throughput. Conversely, a smaller $R$ delays intervention, offering agents more autonomy but limiting the controller’s ability to avoid disorder, such as violations of FIFO order (reversals). This paper emphasizes the importance of determining an optimal control zone radius $R^*$ for balancing objectives of system performance, fairness, and safety.

\section{MODEL FORMULATION} \label{sec:model}

\noindent We formulate the centralized coordination problem as an MILP~\cite{schouwenaars2001mixed,bertsimas2009fairness} that is composed of multiple objective functions $\min\  J = J_{dev}-J_{vel}+J_{rev}$ designed to balance accuracy, throughput, and fairness in managing autonomous agent crossings. Here, $J_{dev}$ penalizes weighted deviations from desired trajectories and terminal states, $J_{vel}$ rewards system throughput (its negative sign ensures that minimizing $J$ is equivalent to maximizing velocity), and $J_{rev}$ encourages fairness by minimizing the total number of priority reversals between agent pairs. This objective is achieved by introducing a binary variable, $O_{q,r}$ for each pair $(q,r)$, which indicates if a reversal has occurred. 

\[
O_{q,r} = 
\begin{cases}
1, & \text{if there is a reversal},\\[4pt]
0, & \text{otherwise }.
\end{cases}
\\ \\
\]

Our centralized control system employs a receding horizon or model predictive control strategy~\cite{garcia1989model, schouwenaars2004receding} where the path of the agents is determined by sequentially solving the MILP for forthcoming time steps, creating finite-horizon solutions that are optimal within each planning window. However, the controller implements only a subset of these computed commands per cycle. Typically, the controller executes only the first set of control inputs, with a fresh set of commands recalculated at each time step~\cite{schouwenaars2001mixed}. The problem may be mathematically stated as

\begin{align*}
\text{minimize} \quad J = \;& 
 \sum_{q=1}^G \left( \sum_{i=1}^{N-1} \alpha^\top \cdot s_{q,i} 
+ \beta^\top \cdot s_{q,N-1} \right) \\
& -  \sum_{q=1}^G \sum_{i=0}^{N-1} \gamma V_{q,i} 
+ \sum_{(q,r)}^G \sum_{i=0}^{N-1} \lambda O_{q,r}
\end{align*}

\quad subject to
% \textbf{{2. Agent dynamics:}}

\[\forall q \in [1,\ldots, G], \;  \forall i \in [0,\ldots, N-1]:\]
\begin{equation}
\begin{aligned}
x_{q,i+1} &= A_{q} \cdot x_{q,i} + B_{q} \cdot u_{q,i} \\
x_{q,0} &= x_{q,\text{current}} \big(p_{q}^{x},\; p_{q}^{y},\; v_{q}^{x},\; v_{q}^{y}\big)^\top,
\end{aligned}
\label{eq:dynamics_constraints}
\end{equation}

% \textbf{{3. Control input limits:}}

\[\forall q \in [1,\ldots, G], \;  \forall i \in [0,\ldots, N-1]:\]
\begin{equation}
\begin{aligned}
u_{q,i} &\geq u_{q,\text{min}} \\
u_{q,i} &\leq u_{q,\text{max}},
\end{aligned}
\label{eq:Control_constraints}
\end{equation}

% \textbf{4. State variable limits:}

\[
\forall q \in [1,\ldots, G],\forall i \in [1,\ldots, N]:
\]
\begin{equation}
\begin{aligned}
x_{q,i} &\geq x_{q,\text{min}} \\
x_{q,i} &\leq x_{q,\text{max}},
\end{aligned}
\label{eq:state_limits}
\end{equation}

\[
\forall q\in [1,\ldots,G],\ \forall i\in [1,\ldots,N]:
\]
\begin{equation}
\begin{aligned}
|p_{q,i}^x - p_{q,f}^x| &\le s_{q,i,1}\\
|p_{q,i}^y - p_{q,f}^y| &\le s_{q,i,2}\\
|v_{q,i}^x - v_{q,f}^x| &\le s_{q,i,3}\\
|v_{q,i}^y - v_{q,f}^y| &\le s_{q,i,4},
\end{aligned}
\label{eq:state_constraints}
\end{equation}

% Safe distance constraint
\[
\forall\, q \in [1,\ldots,G],\;\forall\, r \in [q+1,\ldots,G], \forall\, i \in [1,\ldots,N]:
\]
\begin{equation}
\begin{aligned}
\left| p_{q,i}^x - p_{r,i}^x \right| \geq d_{\text{safe}} \\
\quad \left| p_{q,i}^y - p_{r,i}^y \right| \geq d_{\text{safe}},
\end{aligned}
\label{eq:safe_distance}
\end{equation}

\[
\forall q\in [1,\ldots,G],\ \forall r\in [1,\ldots,G],\ \forall i\in [1,\ldots,N]:
\]
\begin{equation}
\begin{aligned}
\bigl|p_{q,i}^x - c_x\bigr| \le d_{q,i}\\
 \bigl|p_{r,i}^y - c_y\bigr| \le d_{r,i}\\
d_{q,i} + d_{r,i} \ge d_{\mathrm{sep}}\\
d_{q,i}\ge 0,d_{r,i}\ge 0,
\end{aligned}
\label{eq:intersection_dist}
\end{equation}

% Linking positions to crossing indicators using Big-M
\[
\forall q \in [1,\ldots,G], \;\forall i \in [0,\ldots,N-1]:
\]
\begin{equation}
\begin{aligned}
p_{q,i}^x &\ge c_x - M\,(1 - \varpi_{q,i}) \\
p_{q,i}^x &\le c_x - \epsilon + M\,\varpi_{q,i},
\end{aligned}
\label{eq:crossing_q}
\end{equation}

% Linking positions to crossing indicators using Big-M
\[
\forall r \in [1,\ldots,G], \;\forall i \in [0,\ldots,N-1]:
\]
\begin{equation}
\begin{aligned}
p_{r,i}^y &\ge c_y + M\,(1 - \varpi_{r,i}) \\
p_{r,i}^y &\le c_y + \epsilon - M\,\varpi_{r,i},
\end{aligned}
\label{eq:crossing_r}
\end{equation}

\[
\forall q \in [1,\ldots, G], \quad \forall r \in [1,\ldots, G], \quad \forall i \in [0,\ldots, N-1]:
\]
\begin{equation}
\begin{aligned}
\bigl|d_{q,i} - d_{r,i}\bigr| \le \delta_{q,r,i},
\end{aligned}
\label{eq:proximity_calc}
\end{equation}

\[
\forall q \in [1,\ldots, G], \quad \forall r \in [1,\ldots, G], \quad \forall i \in [0,\ldots, N-1]:
\]
\begin{equation}
\begin{aligned}
  \delta_{q,r,i} &\le b + M\,(1 - \mathrm{cl}_{q,r,i}) \\
  \delta_{q,r,i }&\ge b + \epsilon - M\,\mathrm{cl}_{q,r,i},
\end{aligned}
\label{eq:close_proximity}
\end{equation}

% Priority and reversal constraints
\[
\forall (q, r) \mid (r \ne q) \in [1,\ldots, G], \quad \forall i \in [0,\ldots, N-1]:
\]
\begin{equation}
\begin{aligned}
p_{r,i}^y &\ge c_y + \epsilon - M\,\varpi_{q,i} + M\,O_{q,r} + W(1 - \mathrm{cl}_{q,r,i}) \\
p_{q,i}^x &\le c_x - \epsilon + M\,\varpi_{r,i} + M\,O_{q,r} + M(1 - \mathrm{cl}_{q,r,i}).
\end{aligned}
\label{eq:priority_reversal}
\end{equation}

\vspace{1em}

The safe buffer constraint, defined in equation~\eqref{eq:safe_distance}, ensures all agents maintain a minimum distance \(d_{safe}\) at every time step. For each pair \({q},{r}\) on the same flow, it enforces positional separation to be at least \( d_{safe}\) throughout control zone navigation, regardless of which agent is leading or following. The minimum cross-flow separation distance constraint near the intersection, defined in equation~\eqref{eq:intersection_dist}, prevents agents from orthogonal flows (eastbound \({q}\), southbound \({r}\)) from approaching the intersection simultaneously. It requires the sum of their Manhattan distances from \((c_x, c_y)\) to be at least \(d_{sep}\) at each time step during approach or traversal, with the constraint relaxed once both have crossed.

The fairness constraint is designed to minimize reversals and preserve the crossing order for any two distinct agents, $q$ along eastbound and $r$ along southbound at each step $i$ as they approach close to the intersection. Equations~\eqref{eq:crossing_q}–\eqref{eq:priority_reversal} formalize the model with fairness constraints. The variables $\varpi_{q,i}$ and $\varpi_{r,i}$ keep track of whether agents from either flows have crossed the intersection (\({c_{x}}\), \({c_{y}}\)) or not, as shown in equations~\eqref{eq:crossing_q}--\eqref{eq:crossing_r}. The value  $\delta_{q,r,i}$ quantifies, at each step $i$, how close two agents from orthogonal flows are at the intersection. This measure is then used in equation~\eqref{eq:close_proximity} to compare against the specified buffer margin, thereby identifying potential reversal risks when both agents are in proximity ($\mathrm{cl}_{q,r,i}$) at the intersection. {In equation~\eqref{eq:priority_reversal}, the binary variable $O_{q,r}$ indicates whether the controller schedules a reversal between agents $q$ and $r$, taking the value $1$ when their nominal FIFO crossing order is violated and $0$ otherwise.} \\

To ensure safe operation and collision-free trajectories, the control zone must have a sufficient radius. The minimum control zone radius, denoted \(R_{\min}\), defines this feasibility boundary by accounting for braking, separation, and platoon accommodation. This definition applies to a system without fairness constraints. This requirement is formalized in the following proposition.
\begin{Proposition}
Let \(d_{\mathrm{stop}}\) be the worst-case braking distance, \(d_{\mathrm{sep}}\) the minimum cross-flow separation distance, and \(\delta_{\mathrm{platoon}}\) the length of the longest continuous platoon for a given traffic density, $\Lambda$, which has an inter-agent distance as $d_{\mathrm{safe}}$. Then the minimum control zone radius is given by:
\begin{align}
&R_{\min} = d_{\mathrm{stop}} + d_{\mathrm{sep}} + \delta_{\mathrm{platoon}}, \\
&\delta_{\mathrm{platoon}} = (n_{\mathrm{platoon}} - 1)d_{\mathrm{safe}}.
\end{align}
where $n_{\mathrm{platoon}}=\max\{n_1,n_2\}$. $n_1$ and $n_2$ represent platoon sizes in each flow, and the inter-agent distance is $d_{\mathrm{safe}}$. \\
The bound $R_{\min}$ ensures feasibility and guarantees collision-free trajectories. The optimal radius $R^*$ satisfies $R^* \ge R_{\min}$.
\end{Proposition}
\begin{proof}
If $R < R_{\min}$, then at least one requirement, stopping distance, cross-flow separation, or platoon accommodation, cannot be satisfied. In this case, collision-free trajectories are not guaranteed. If $R \geq R_{\min}$, the control zone can accommodate all three requirements simultaneously, ensuring safe operation. Since feasibility is a prerequisite for optimality, the optimal radius must satisfy $R^* \geq R_{\min}$.
\end{proof}
\begin{remark}
In the simulation, the control zone radius is not initialized at 
$R_{\min}$, as larger values are typically chosen to capture better and analyze performance trends. Moreover, when fairness constraints are imposed, the minimum radius increases, as additional space is required to accommodate equitable flow scheduling.
\end{remark}
\section{Fairness Measure}\label{sec:fairness-measure}
To assess how fair a centralized controller is, we compare the order in which the agents enter the control zone with the order in which they cross the intersection, as shown in Fig.~\ref{fig:fairness_order} and in Table~\ref{tab:crossing_orders}. If the agent maintains its relative position, i.e., no agent from the orthogonal flow resequenced ahead of it, and the FIFO rule is preserved, then this order is considered fair; otherwise, such a resequencing event is counted as a reversal. The following equations~\eqref{eq:order_shift}--\eqref{eq:zero_sum} provide a concise metric basis for assessing the fairness of the centralized controller through reversal counts and order shifts.\\

\paragraph{Order shift (Table~\ref{tab:crossing_orders_example})}
For agent~$q$,
\begin{equation}
\text{Order Shift} \;=\;
\text{Arrival Order} \;-\;
\text{Crossing Order}
\label{eq:order_shift}
\end{equation}
where \(\text{Arrival Order}\) is the order of agent~\(q\) when it first enters the control zone, and \(\text{Crossing Order}\) is the order of agent~\(q\) when it exits the intersection.

\paragraph{Reversal Indicator}
For agent $q$,
\begin{equation}
\text{Reversal Indicator} \;=\;
\begin{cases}
1, & \text{if } \text{order shift} \neq 0,\\[4pt]
0, & \text{otherwise}.
\end{cases}
\label{eq:reversal_indicator}
\end{equation}
An example is illustrated in Table~\ref{tab:crossing_orders_example}. 

\paragraph{Zero‑sum block}
Because of the reversal, a platoon forms near the intersection. The relative order of the agents is swapped at the cross flows, leading to positive and negative shifts that cancel out in aggregate within a contiguous block, as illustrated in Table~\ref{tab:crossing_orders}. The total order shift is zero, as shown in the following equation:
\begin{equation}
\sum_{q=1}^{G} \text{(Order Shift)}_q \;=\; 0 ,
\label{eq:zero_sum}
\end{equation}
where $G$ is the total number of agents considered.\\

\begin{table}[ht]
\centering
\caption{Agent Arrival order vs.\ Intersection crossing order and flow direction at the intersection.}
\label{tab:crossing_orders}
\begin{tabular}{ccc c}
\toprule
\textbf{Arrival Order} & \textbf{Crossing Order} & \textbf{Flow Direction} & \textbf{Order shift}  \\
\midrule
 1  & 1  & x & 0 \\
 2  & 2  & x & 0 \\
 3  & 3  & x & 0 \\
 4  & 4  & x & 0 \\
 5  & 5  & y & 0 \\
 6  & 7  & y & -1 \\
 7  & 9  & y & -2 \\
 8  & 6  & x & 2 \\
 9  & 8  & x & 1 \\
10  & 10 & x & 0 \\
11  & 11 & x & 0 \\
\bottomrule
\end{tabular}
\end{table}

\begin{table}[h]
\centering
\caption{Illustrative order-shift examples from Table~\ref{tab:crossing_orders}.}
\label{tab:crossing_orders_example}
\begin{tabular}{@{}cccc@{}}
\toprule
Agent & Arrived / Crossed & Order Shift & Reversal \\
\midrule
2 & $2^{\text{nd}}$ / $2^{\text{nd}}$ & $2-2=0$  & 0 \\
6 & $6^{\text{th}}$ / $7^{\text{th}}$ & $6-7=-1$ & 1 \\
8 & $8^{\text{th}}$ / $6^{\text{th}}$ & $8-6=+2$ & 1 \\
\bottomrule
\end{tabular}
\end{table}

As reversals represent pairwise swaps, the shifts would cancel within the block. For the example in Table~\ref{tab:crossing_orders}, the zero-sum block consists of the order shifts from Agents 6, 7, 8, and 9: \((-1) + (-2) + (+2) + (+1) = 0\). This block is defined by a change in flow direction, starting after Agent 5’s \(y\)-flow and ending before Agent 10’s \(x\)-flow. This swap pattern causes the zero-sum block to align with the same flow direction as its neighboring agents, as shown in Table~\ref{tab:crossing_orders}. If the aligned flow directions are the same, the platoon counts are updated as \( \text{platoons\_y} = 1 \) and \( \text{platoons\_x} = 1 \). We track only the formation event, not the platoon size.   

\begin{figure}[!ht]
    \centering
    \includegraphics[width=\columnwidth]{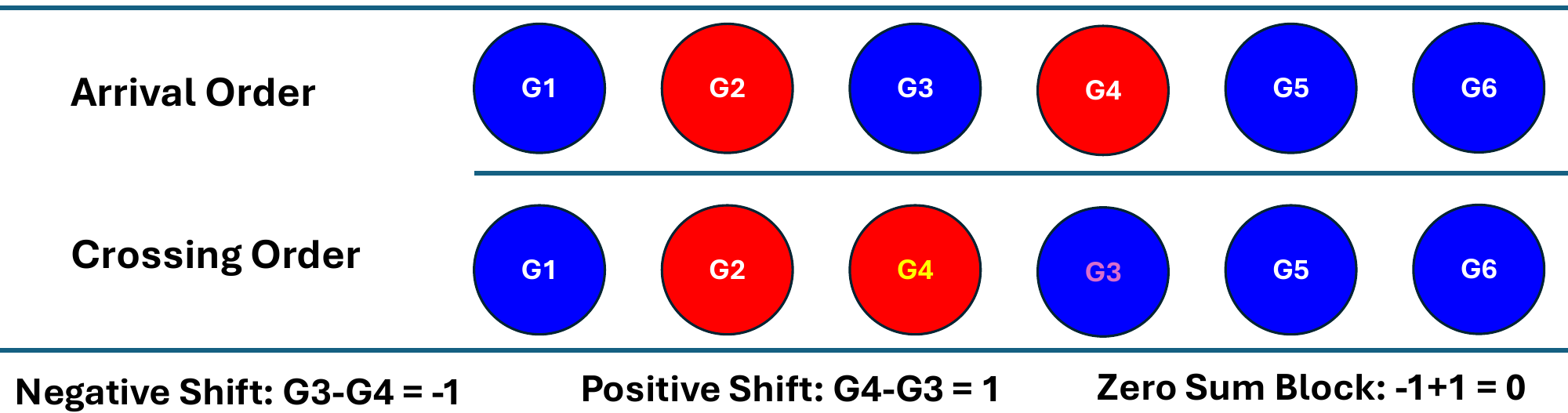}
    \captionsetup{skip=6pt} % space between image and caption
    \caption{Illustration of shifts or swapping in intersection crossing order.}
    \label{fig:fairness_order}
\end{figure}
\vspace{-14pt} % reduce space before paragraph starts

\subsection{Why Reversal Happens.}
The simulation scenario has two orthogonal flows where agents enter the control zone at the maximum velocity in their arrival order. Reversals occur when agents from orthogonal flows reach the intersection nearly simultaneously, and the centralized controller must prioritize one agent to maintain a minimum cross-flow separation distance, $d_{\mathrm{sep}}$, forcing the other to wait. This reordering changes the agent’s crossing orders relative to arrival orders equation~\eqref{eq:order_shift}, which is then captured by reversal indicators equation~\eqref{eq:reversal_indicator}. Reversals introduce localized delays, which propagate backward along the flow, causing trailing agents to cluster into platoons. Consequently, reversals act as a natural mechanism for platoon formation.

%While reversals undermine fairness by disrupting the nominal crossing orders, they may enhance efficiency by producing denser platoons with reduced inter-agent spacing, thereby increasing overall throughput. 
While reversals undermine fairness by disrupting the nominal crossing orders, they may enhance efficiency. They produce denser platoons with reduced inter-agent spacing, thereby increasing overall throughput. Hence, an inherent trade-off emerges where minimizing reversals improves fairness but limits the throughput gains from platoon formation.

\section{Simulation Framework}\label{sec_simulation}

% \paragraph{Scenario and Geometry.}
We consider two orthogonal, single-lane flows that intersect at the origin. The north--south flow runs from $y=+240~\mathrm{m}$ to $y=-240~\mathrm{m}$, and the west--east flow runs from $x=-240~\mathrm{m}$ to $x=+240~\mathrm{m}$. A circular control zone of radius $R$ is centered at the meeting point $(0,0)$ of two flows. The simulation is divided into two phases: the first phase is free-flow, where agents outside the control zone travel at constant speed along their lane, and their positions are updated by forward-Euler discretization of the kinematics. In the second phase, the agents, once inside the control zone, are coordinated by the centralized controller MILP in Section \ref{sec:model}.

Agents on each flow arrive according to independent Poisson processes with rates $\Lambda_x$ (x-flow) and $\Lambda_y$ (y-flow), specified in Table~\ref{tab:sim_settings}. Equivalently, the continuous inter-arrival times are exponential,

\[ \Delta t_j \sim \mathrm{Exp}(\Lambda_\ell),\quad \ell\in\{x,y\}, \]

\noindent which we discretize to simulation steps of size $\Delta t$ by $k_j=\lceil \Delta t_j/\Delta t \rceil$. To prevent unrealistically short headways, we enforce a minimum interarrival spacing of $k_{\min}$ steps (e.g., $k_{\min}=2$ steps corresponds to $0.2~\mathrm{s}$\,) between successive arrivals on each flow, updating $k_j \leftarrow \max(k_j,\,k_{\min})$. With a fixed random seed ($s=20$), the arrival schedule is reproducible within the Python environment. 

\begin{table}[!t]
  \centering
  \caption{Simulation parameters and solver settings.}
  \label{tab:sim_settings}
  \resizebox{\columnwidth}{!}{%
  \begin{threeparttable}
  \begin{tabular}{@{}ll@{}}
    \toprule
    \textbf{Geometry \& traffic} & \\
    Control zone radius $R$ & $[40,60,80,100, 120,140,160,180,200]$ m \\
    Safe distance $d_{\text{safe}}$ & $3$ m \\
    Separation distance $d_{\text{sep}}$ & $4$ m  \\
    Start/End (N–S lane) & $+240$ m $\rightarrow$ $-240$ m; y-flow \\
    Start/End (W–E lane) & $-240$ m $\rightarrow$ $+240$ m; x-flow \\
    Arrival rates $(\Lambda_x,\Lambda_y)$ & $[0.5,0.7,0.9]$ s$^{-1}$ each; $k_{\min}$ $=0.2$ s (2 steps) \\
    \midrule
    \textbf{Discretization \& horizon} & \\
    Time step $\Delta t$ & $0.1$ s \\
    Horizon length $N$ & $60$ steps \\
    Simulation duration $T$ & $3600$ s \\
    \midrule
    \textbf{Bounds} & \\
    % Max speed $V_{\max}$ & $15$ m/s \\
    % Max acceleration $A_{\max}$ & $3$ m/s$^2$ \\
    Input bounds $u_{q,i}$ & $[-5,\,5]$ m/s$^2$ (per flow) \\
    % State bounds $x_{q,i}$ &
    %   $p^x\in[-(R\!+\!2),\,R\!+\!29]$ m,\ $p^y\in[-(R\!+\!29),\,R\!+\!2]$ m; \\
    %   & $v^x,v^y\in[-15,\,15]$ m/s \\
    State bounds $x_{q,i}$ &
      $p^x\in[-R,\,R]$ m,\ $p^y\in[-R,\,R]$ m; \\
      & $v^x,v^y\in[-15,\,15]$ m/s \\
    \midrule
    \textbf{Objective weights} & \\
    State slacks $\alpha$ & $[1,1,1,1]$ \\
    Terminal slacks $\beta$ & $[100,100,100,100]$ \\
    Speed reward $\gamma$ & $1$ \\
    Reversal penalty $\lambda$ & $1$ \\
    \midrule
    \textbf{MIP constants} & \\
    Big-$M$, Big-$W$ & $10^3,\ 10^6$ \\
    Proximity buffer $b$ & $4$ m \\
    Strictness $\varepsilon$ & $0.1$ m \\
    \midrule
    \textbf{Solver/Hardware} & \\
    Solver & Gurobi Optimizer (Python API) \\
    CPU & Intel(R) Xeon(R) Gold 6248 CPU 2.50GHz \\
    \bottomrule
  \end{tabular}
  \end{threeparttable}
  }% end resizebox
\end{table}

The centralized controller activates when the first agent enters the control zone from either flow. At each simulation step, it identifies all agents currently within $[-R,R]\times[-R,R]$, creates their decision variables and constraints, and solves a receding-horizon MILP of length $N$ steps (corresponding to a horizon duration of $N$×$\Delta t$ seconds) (see Section~\ref{sec:model}). Only the first set of control inputs is executed,  with a fresh set of commands recalculated at each time step. Agents that exit the zone are removed from the MILP and resume free-flow kinematics. The complete trajectory, a log of time-stamped positions, and optimized decisions of each agent are stored in a JSON file, which is used for post-processing and result generation. This simulation study serves as a rigorous platform to validate the proposed methodology’s applicability, robustness, and potential.

The simulation recording and MILP code are available in \href{https://github.com/salmansarfarazghori/Agent_control_region_circle/blob/main/README.md}{GitHub Repository README} or \href{https://www.youtube.com/watch?v=mqgiasBgRtE}{video link}.

\section{Results} \label{results}
% \subsection{System performance metric }
\subsection{System Performance without Fairness Constraints}\label{subsec:sys_perf_false}
To evaluate the performance of the control zone, we use total delay and total energy as metrics. The delay is measured as the actual time it takes in the simulation for an agent to cross the entire control zone of radius $R$ against the ideal time the agent takes to travel with maximum velocity to a given radius. It is defined as

\begin{equation}\label{eq:delay_main}
\mathrm{delay}^{(q)}(R) = t^{(q)}_{\mathrm{sim}}(R) - t^{(q)}_{\mathrm{ideal}}(R).
\end{equation}

\noindent The energy consumption is measured as the cumulative change in velocity of the agent, throughout its journey inside $R$. It is defined as

\begin{equation}
E_{q} \;=\; 
\sum_{} 
\bigl|\,v_{q,i+1} \;-\; v_{q,i}\bigr|.
\label{eq:proxy_energy}
\end{equation}

\noindent The definition in equations~\eqref{eq:delay_main}--\eqref{eq:proxy_energy} is computed over the full entry–exit path of length $2R$, since agents travel at $V_{\max}$ outside the control zone. This ensures that any observed delay is solely due to the controller’s actions, enabling unbiased comparisons across different values of $R$. 

The study of system performance under different operational strategies, with and without fairness constraints, is illustrated in Figs.~\ref{fig:delay_3} and~\ref{fig:delay_7}. The Tables~\ref{tab:total_delay_full_common}--\ref{tab:total_delay_full_common_fair} summarize the total delays and energy consumption across varying traffic densities. These results provide a detailed look at the performance trends when fairness is inactive. This case will be discussed in the following section.

Fig.~\ref{fig:delay_3} represents the normalized total delay calculated using equation~\eqref{eq:delay_main} across various control zone radii for the different traffic densities ($\Lambda = 0.9, 0.7, 0.5~\mathrm{s^{-1}}$). The normalization employs a min-max approach, highlighting the relative variation and trend rather than absolute magnitudes. Across all densities, the curves exhibit a common pattern where total delay decreases sharply as the control radius increases from $40~\mathrm{m}$ up to $90~\mathrm{m}$ within which the optimal point is attained. The radius at which the delay reaches its minimum is referred to as the optimal point, since it delivers the best system performance in terms of delay reduction. This behavior indicates that moderate enlargement of the control zone improves the system performance. Beyond this optimal point, the gains in total delay saturate and even slightly increase for larger control radii.

The analysis across traffic densities reveals that the optimal control zone radius shifts to larger values as density increases, as summarized in Table~\ref{tab:total_delay_full_common}. Specifically, for lower densities ($\Lambda = 0.5~\mathrm{s^{-1}}$), the optimal radius is around $60$~m, it shifts to near $80~\mathrm{m}$ for moderate density ($\Lambda = 0.7~\mathrm{s^{-1}}$) and reaches around $90~\mathrm{m}$ for high density ($\Lambda = 0.9~\mathrm{s^{-1}}$). This progressive shift highlights that denser traffic requires a larger spatial area for efficient centralized management and improves system performance.\\

\begin{figure}[!ht]
    \centering
    \includegraphics[width=\columnwidth]{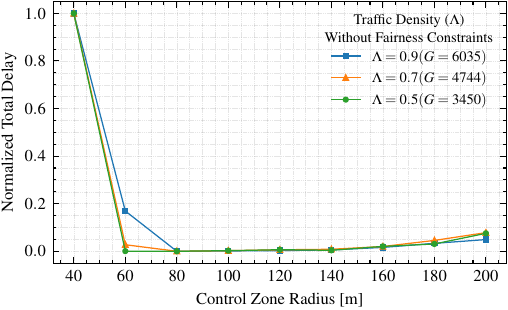}
    \captionsetup{skip=3pt}
   \caption{Total delay of all the agents evaluated over the entire control zone path length ($2R$) around the intersection for different traffic densities for the system without fairness constraint vs varying control zone radius.}
    \label{fig:delay_3}
\end{figure}

\begin{table}[t]
  \centering
  \caption{Total delay [s] by control zone radius and traffic density, for the system without fairness constraints.}
  \label{tab:total_delay_full_common}
  \begin{tabular}{@{}cccc@{}}
    \toprule
    \textbf{Radius $R$ [m]} & $\boldsymbol{\Lambda=0.9~\mathrm{s^{-1}}}$ & $\boldsymbol{\Lambda=0.7~\mathrm{s^{-1}}}$ & $\boldsymbol{\Lambda=0.5~\mathrm{s^{-1}}}$ \\
    \midrule
    40  & 340.742209 & 181.804367 &  81.682809 \\
    60  & 311.017231 & 164.473746 &  74.407344 \\
    80  & 304.958823 & 163.983347 &  74.402705 \\
    100 & 305.005939 & 164.027679 &  74.421047 \\
    120 & 305.059639 & 164.094644 &  74.450681 \\
    140 & 305.165351 & 164.133742 &  74.433219 \\
    160 & 305.540043 & 164.339765 &  74.552552 \\
    180 & 306.150742 & 164.796784 &  74.626199 \\
    200 & 306.718879 & 165.377634 &  74.951762 \\
    \bottomrule
  \end{tabular}
\end{table}

Fig.~\ref{fig:delay_5} represents the normalized total energy consumption, calculated using equation~\eqref{eq:proxy_energy} across different control zone radii for traffic densities ($\Lambda = 0.9, 0.7, 0.5~\mathrm{s^{-1}}$). The total energy consumption decreases sharply as the radius increases from $R=40~\mathrm{m}$ to around $R=90~\mathrm{m}$, within which the optimal point is attained. This reduction reflects smoother velocity adjustments enabled by earlier intervention, represented by a larger control zone. Beyond this optimal point, the improvements in energy consumption saturate and eventually begin to increase slightly. This saturation pattern is consistent with the delay results in Fig.~\ref{fig:delay_3}, indicating the presence of an optimal control radius for each traffic density. Beyond this point, enlarging the control zone provides no further benefit. The shift of the optimal radius toward larger values with increasing traffic density is summarized in Table~\ref{tab:total_energy_full_common}.

These findings collectively indicate the existence of an optimal control zone radius, highlighting the need to adapt the control zone radius $R$ to traffic density to maintain efficient operation. Simulation results show that as $R$ increases from a small value, system performance improves. However, these gains gradually diminish and eventually saturate. This behavior suggests that beyond a certain point, enlarging $R$ provides no further benefit.

\begin{remark}
    The existence of an optimal radius was first reported empirically in our earlier work~\cite[Hypothesis~1]{11108008} for single traffic density. The consistent saturation pattern observed across multiple densities in this study strengthens that observation and motivates its formal statement as a hypothesis. Intuitively, once $R$ is large enough to enable all feasible schedules for the given density, additional space only increases the number of agents being managed simultaneously, without improving performance. A rigorous proof is left for future work, as enlarging $R$ changes both the set of feasible trajectories and the coupling between agents, which may change the optimal solution in nontrivial ways.
\end{remark}

\begin{figure}[!ht]
    \centering
    \includegraphics[width=\columnwidth]{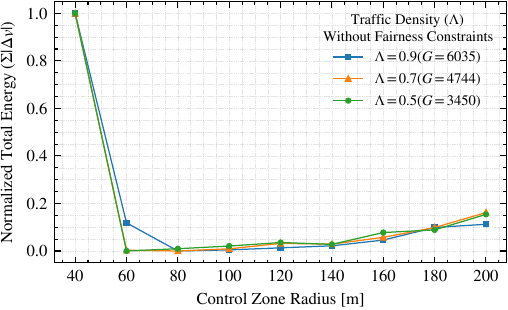}
    \captionsetup{skip=3pt}
   \caption{Total energy of all the agents evaluated over the entire control zone path length ($2R$) around the intersection for different traffic densities for the system without fairness constraint vs varying control zone radius.}
    \label{fig:delay_5}
\end{figure}

\begin{table}[t]
  \centering
  \caption{Total consumed energy (\( \sum \lvert \Delta v \rvert \)) by control zone radius and traffic density, for the system without fairness constraints.}

  \label{tab:total_energy_full_common}
  \begin{tabular}{@{}cccc@{}}
    \toprule
    \textbf{Radius $R$ [m]} & $\boldsymbol{\Lambda=0.9~\mathrm{s^{-1}}}$ & $\boldsymbol{\Lambda=0.7~\mathrm{s^{-1}}}$ & $\boldsymbol{\Lambda=0.5~\mathrm{s^{-1}}}$ \\
    \midrule
    40  & 11742.805860 & 6616.046832 & 3155.093870 \\
    60  & 10671.814390 & 5930.174859 & 2877.317466 \\
    80  & 10528.061849 & 5928.066784 & 2879.969256 \\
    100  & 10533.671732 & 5934.325602 & 2883.067891 \\
    120  & 10544.047833 & 5950.032519 & 2887.226340 \\
    140  & 10554.425725 & 5948.001537 & 2885.028172 \\
    160  & 10583.870933 & 5967.436205 & 2898.879083 \\
    180  & 10648.101256 & 5995.982287 & 2901.951909 \\
    200  & 10664.629999 & 6039.258762 & 2920.092437 \\
    \bottomrule
  \end{tabular}
\end{table}

\subsection{System Performance with Fairness Constraints}

Without fairness constraints, the two performance metrics, total delay and total energy consumption, both supported the existence of an optimal control zone radius (see Figs.~\ref{fig:delay_3} and~\ref{fig:delay_5}). In contrast, enabling the fairness constraint modifies system dynamics very much. As shown in Figs.~\ref{fig:delay_6} and~\ref{fig:delay_7}, delay and energy consumption decrease when the radius increases till $80~\mathrm{m}$ or $100~\mathrm{m}$, and beyond the $100~\mathrm{m}$ performance metrics worsen substantially with larger radii, suggesting that the fairness constraint significantly degrades agent coordination. This trend occurs at all traffic densities ($\Lambda = 0.9, 0.7, 0.5~\mathrm{s^{-1}}$), where the fairness constraint’s implementation leads to high coordination complexities as the radius increases to larger values.

\begin{table}[t]
  \centering
  \caption{Total delay[s] by control zone radius and traffic density, for the system with fairness constraints.}
  \label{tab:total_delay_full_common_fair}
  \begin{tabular}{@{}cccc@{}}
    \toprule
    \textbf{Radius $R$ [m]} & $\boldsymbol{\Lambda=0.9~\mathrm{s^{-1}}}$ & $\boldsymbol{\Lambda=0.7~\mathrm{s^{-1}}}$ & $\boldsymbol{\Lambda=0.5~\mathrm{s^{-1}}}$ \\
    \midrule
    40  & \textemdash & \textemdash & \textemdash \\
    60  & 331.610227 & 175.109843 & 77.822457 \\
    80  & 328.477115 & 174.979778 & 77.818686 \\
    100 & 328.567348 & 175.046228 & 77.818677 \\
    120 & 329.092215 & 175.197285 & 77.853453 \\
    140 & 329.716541 & 175.601019 & 78.040673 \\
    160 & 330.347674 & 175.996415 & 78.173349 \\
    180 & 330.634515 & 176.272579 & 78.274321 \\
    200 & 330.790821 & 176.401174 & 78.388164 \\
    \bottomrule
  \end{tabular}
\end{table}

\begin{figure}[!ht]
    \centering
    \includegraphics[width=\columnwidth]{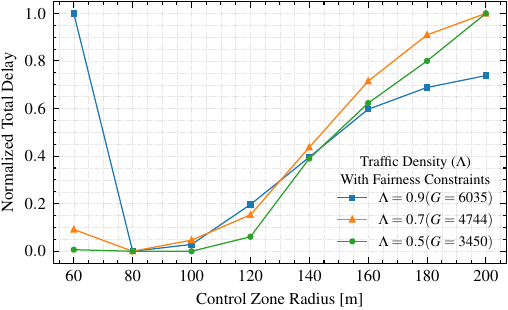}
    \captionsetup{skip=3pt}
   \caption{Total delay of all the agents evaluated over the entire control zone path length ($2R$) around the intersection for different traffic densities for the system with fairness constraint vs varying control zone radius.}
    \label{fig:delay_6}
\end{figure}

\begin{figure}[!ht]
    \centering
    \includegraphics[width=\columnwidth]{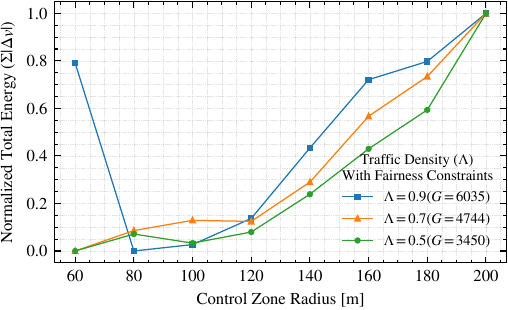}
    \captionsetup{skip=3pt}
   \caption{Total Energy consumption of all the agents evaluated over the entire control zone path length ($2R$) around the intersection for different traffic densities for the system with fairness constraint vs varying control zone radius.}
    \label{fig:delay_7}
\end{figure}

\begin{remark}
    While fairness constraints support equitable treatment among agents crossing order, they impose significant performance penalties. Therefore, determining an optimal control zone radius becomes vital for balancing fairness with overall system performance. This understanding is important for centralized traffic management systems and emphasizes the importance of considering fairness explicitly in the design and implementation of coordination algorithms for autonomous agent management.
\end{remark}

\subsection{MILP Optimality}
The observed increase in total delay and energy consumption in Section~\ref{subsec:sys_perf_false} for control zone radii beyond $R=140~\mathrm{m}$, which occurs well after the optimal zone (around $R=80-90~\mathrm{m}$) and the subsequent saturation plateau as shown in Figs.~\ref {fig:delay_3} and~\ref{fig:delay_5}. This raises the possibility that increasing computational complexity, reflected in solver performance, may be a contributing factor. To investigate this, we examine the final Mixed Integer Programming (MIP) gap given by the Gurobi solver~\cite{gurobi} for every MPC function call. The results are visualized as violin density plots in Figs.~\ref{fig:MIP_GAP_WITHOUTFAIR} and~\ref{fig:MIP_GAP_WITHFAIR}. In these violins, the width encodes the observed density of MIP gaps at a given radius, and wider sections indicate more runs accumulated near that value.

Across both cases, whether fairness constraints are imposed or not, the solver performance remains stable. While there is a slight broadening of the tails as $R$ grows, the solver overall maintains tight MIP gaps. The final MIP gap, which represents the relative difference between the best-found solution and the theoretical optimal solution, never exceeds $0.011\%$,  indicating that the solution found is near-optimal regardless of the control zone radius. In MILP practice, it is within numerical optimality tolerance. This observation clarifies that the increase in total delay and energy at larger radii is not caused by numerical difficulties or solver instability.

Instead, the performance degradation at large $R$ originates from the coordination dynamics of the problem itself. Enlarging the control zone necessitates the simultaneous scheduling of more agents, while satisfying operational and fairness constraints, which substantially narrows the feasible solution space, even when MILP solutions are near-optimal. The consistently low MIP gaps in Figs.~\ref{fig:MIP_GAP_WITHOUTFAIR} and~\ref{fig:MIP_GAP_WITHFAIR} confirm that the observed results arise from fundamental coordination limitations rather than computational factors.

\begin{figure}[!ht]
    \centering
    \begin{overpic}[width=\columnwidth]{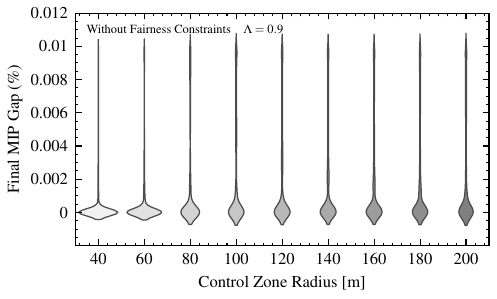}
    \end{overpic} 
    \caption{The kernel-density distributions of the final MIP gaps at each control zone radius $R$ without fairness constraints.}
    \label{fig:MIP_GAP_WITHOUTFAIR}
\end{figure}
\begin{figure}[!ht]
    \centering
    \begin{overpic}[width=\columnwidth]{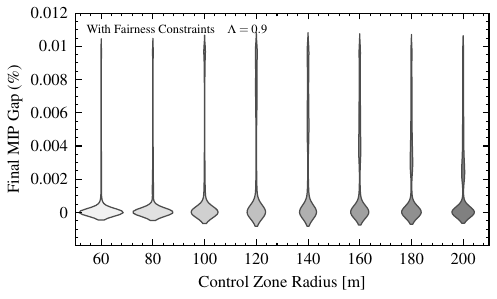}
    \end{overpic} 
    \caption{The kernel-density distributions of the final MIP gaps at each control zone radius $R$ with fairness constraints.}
    \label{fig:MIP_GAP_WITHFAIR}
\end{figure}
\subsection{Reversals, Platooning, and System efficiency}
The analysis presented in Table~\ref{tab:fairness-comparison} reports the impacts of fairness constraints inactive versus active in the centralized controller (see ~\ref{sec:model} ) across control zone radii for traffic density of $\Lambda=0.9~\mathrm{s}^{-1}$. Notably, the table highlights the trade-offs between maintaining fairness, which is defined through minimizing crossing order reversals, and the potential efficiency gains achieved through platoon formation at intersection crossings by varying the control zone radius.

In the system without fairness constraints, there are significant crossing order deviations as shown in Table~\ref{tab:fairness-comparison} (left block), which are measured using the ‘Sum’ metric, a sum of absolute order shifts equation~\eqref{eq:order_shift}, and by reversal counts equation~\eqref{eq:reversal_indicator} that occur across all control zone radii. For example, at $ R=40~\mathrm {m}$ for a system without fairness constraints, register 616 reversals, showing a considerable deviation from the FIFO crossing order rule. Interestingly, these reversals assist in forming platoons calculated by equation~\eqref{eq:zero_sum}, which shows a count of $182$ and $162$ platoons for X and Y flows (see Section~\ref{sec:fairness-measure}), respectively. This platooning naturally reduces spacing gaps between successive agents, thereby improving overall throughput and operational efficiency. As the control zone radius increases from $R=40~\mathrm{m}$ to $R=80~\mathrm{m}$, reversal counts decrease while platoon counts rise, indicating improved fairness and operational efficiency without suppressing beneficial agent grouping.

\noindent When compared to $R=80~\mathrm{m}$, reversal counts show a slight uptick for $R=100$ to $160~\mathrm{m}$, but remain roughly the same across this range, and the platoon count also remains nearly the same, indicating that near-optimal platoon structures are sustained. Beyond $R=180~\mathrm{m}$, both reversals and platoon counts increase, signaling rising coordination overhead and a resulting loss of efficiency. These trends align with Fig.~\ref{fig:delay_3} (blue curve $\Lambda=0.9~\mathrm{s}^{-1}$) where lower total delay for $R=40$ to $80~\mathrm{m}$ (fewer reversals, stronger platooning), delay saturation for $R=100$ to $160~\mathrm{m}$, and increased delay beyond $R=180~\mathrm{m}$ due to higher coordination complexity.

% Define Python plot colors
\definecolor{plotblue}{HTML}{1f77b4}  % blue
\definecolor{plotorange}{HTML}{ff7f0e} % orange

\begin{table*}[ht]
  \centering
  \caption{Effect of fairness due to crossing-order deviations and local platooning across control zone radii.}
  \label{tab:fairness-comparison}
  \sffamily
  \newcommand{\fcol}{\textcolor{plotblue}}   % Baseline (No Fairness)
  \newcommand{\tcol}{\textcolor{plotorange}} % With Fairness
  \begin{tabular}{@{}c
                  cccc     
                  r@{\hspace{1.5em}} 
                  cccc     
                 @{}}
    \toprule
    & \multicolumn{4}{c}{\bfseries Without Fairness Constraints} &
      & \multicolumn{4}{c}{\bfseries With Fairness Constraints} \\
    \cmidrule{2-5} \cmidrule{7-10}
    \textbf{Radius [m]} &
    \fcol{Sum $\Delta$} & \fcol{Reversals} & \fcol{Plat.\,X} & \fcol{Plat.\,Y} &
    & \tcol{Sum $\Delta$} & \tcol{Reversals} & \tcol{Plat.\,X} & \tcol{Plat.\,Y} \\
    \midrule
     40  & 700 & 616 & 182 & 162 & & –   & –   & –   & –  \\
     60  & 682 & 588 & 181 & 176 & & 300 & 262 & 106 & 78 \\
     80  & 668 & 572 & 184 & 183 & & 310 & 261 & 108 & 75 \\
    100  & 676 & 577 & 185 & 183 & & 310 & 261 & 108 & 75 \\
    120  & 676 & 579 & 184 & 183 & & 310 & 261 & 108 & 75 \\
    140  & 672 & 575 & 187 & 184 & & 304 & 258 & 108 & 75 \\
    160  & 672 & 578 & 188 & 183 & & 302 & 257 & 108 & 74 \\
    180  & 690 & 598 & 192 & 186 & & 302 & 256 & 107 & 74 \\
    200  & 714 & 615 & 198 & 185 & & 300 & 255 & 107 & 74 \\
    \bottomrule
  \end{tabular}
\end{table*}

\begin{figure}[!ht]
    \centering
    \begin{overpic}[width=\columnwidth]{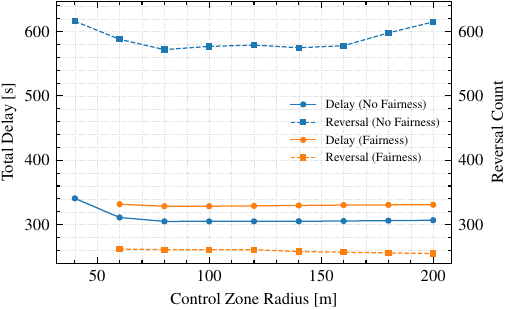}
    \end{overpic} 
    \caption{Total delay and reversal counts for traffic density $\Lambda=0.9~\mathrm{s^{-1}}$ for varying control zone radius with and without Fairness constraint.}
    \label{fig:delay_reversal_false_combine}
\end{figure}

Activating the fairness constraint markedly reduces the order deviation, see Table~\ref{tab:fairness-comparison} (right block), but raises the delay compared to the system without fairness constraints, as shown in Fig.~\ref{fig:delay_reversal_false_combine}. Reversals drop by roughly one-half for the control zone from $R=60~\mathrm{m}$ to $R=200~\mathrm{m}$ compared to the system without fairness constraints. The results demonstrate that fairness constraints work very well in reducing the reversals and making the system fairer. Minimum control zone for the system without fairness is around $R=40~\mathrm{m}$, but now, when the fairness constraint is active, the minimum $R$ is shifted to around $R=60~\mathrm{m}$ as respecting the FIFO crossing order rule needs more space to avoid overly restrictive schedules. Beyond $R=120~\mathrm{m}$, reversal counts decrease slightly while total delay increases (Fig.~\ref{fig:delay_6}, blue curve, $\Lambda=0.9~\mathrm{s}^{-1}$). This occurs because prioritizing order preservation (fewer reversals) together with scheduling a larger number of agents in the expanded control zone tightens feasibility and raises coordination overhead, thereby increasing delay.

Across all densities, the total delay of the fairness active controller is higher when compared to the fairness inactive controller, and this is expected as fewer reversal means fewer platooning formations, which means less improvement in system efficiency as shown in Fig.~\ref{fig:delay_reversal_false_combine}. So there is a trade-off between the controller being fairer and system performance as illustrated by Tables~\ref{tab:total_delay_full_common}-~\ref{tab:total_delay_full_common_fair}. 

\subsection{Effect of Control Zone Radius on Velocity Profiles and System Behavior}

The choice of control zone radius is a key factor for effective planning. A small control zone (e.g., $R=40~\mathrm{m}$ to $R=60~\mathrm{m}$) provides limited space for planning, abrupt velocity changes (upper panel in Fig.~\ref{fig:Velocity_1}) leading to inefficient transient behaviors and higher delay and greater energy consumption. In contrast, a moderately sized control zone enables smooth coordination, resulting in steady-state velocities (lower panel in Fig.~\ref{fig:Velocity_1}, and both panels in Fig.~\ref{fig:Velocity_2}) and beneficial platoon formation. A very large zone (beyond $R=140~\mathrm{m}$) introduces coordination overhead by managing more agents, resulting in performance saturation despite near-optimal solving. Selecting an appropriate control zone radius improves overall system efficiency.

\begin{figure*}[h!]
    \centering
    \begin{overpic}[width=\textwidth]{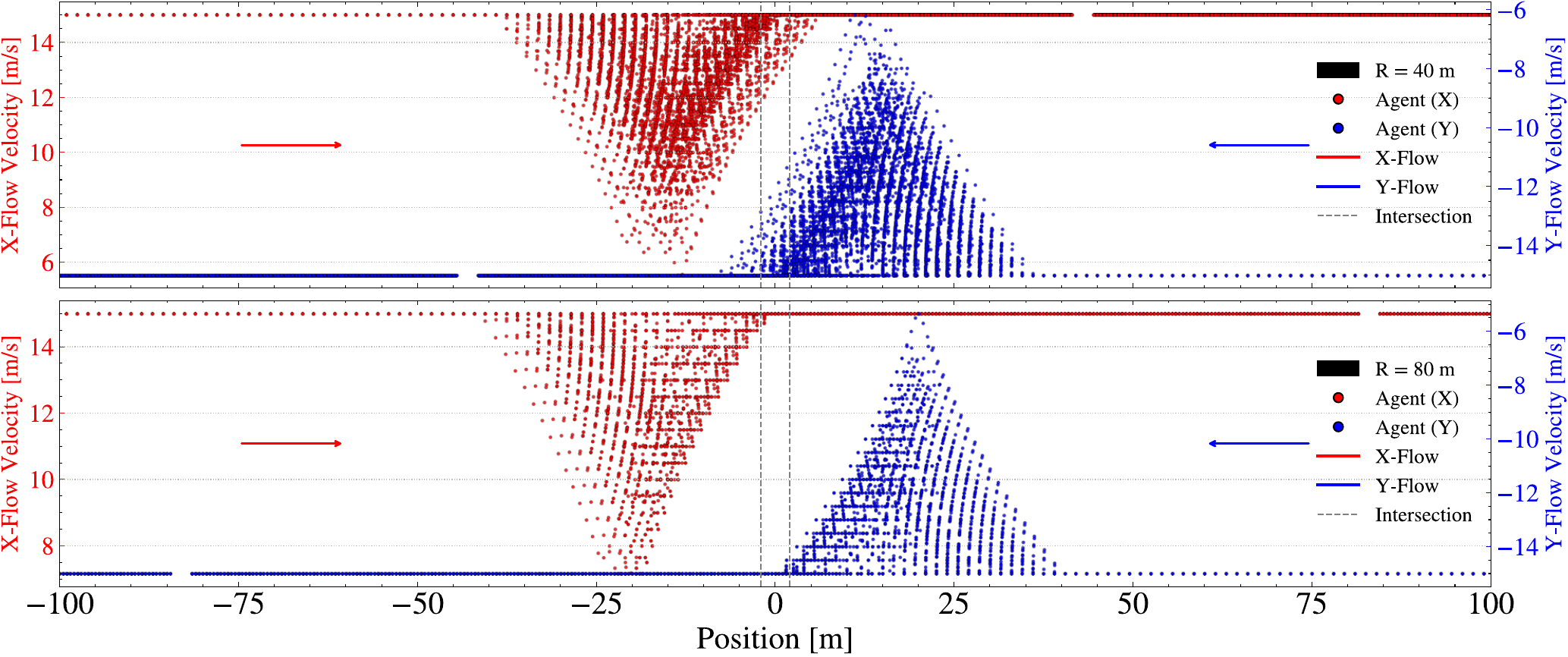}
    \end{overpic} 
    \captionsetup{skip=3pt}
    \caption{Velocity profile for control zone radius of $R=40~\mathrm{m}$ (top plot) and $R=80~\mathrm{m}$ (bottom plot) for traffic density $\Lambda=0.9~\mathrm{s^{-1}}$, for the system without fairness constraint. Red and blue dots denote agents traveling along the x- and y-flows.}
    \label{fig:Velocity_1}
\end{figure*}

\begin{figure*}[h!]
    \centering
    \begin{overpic}[width=\textwidth]{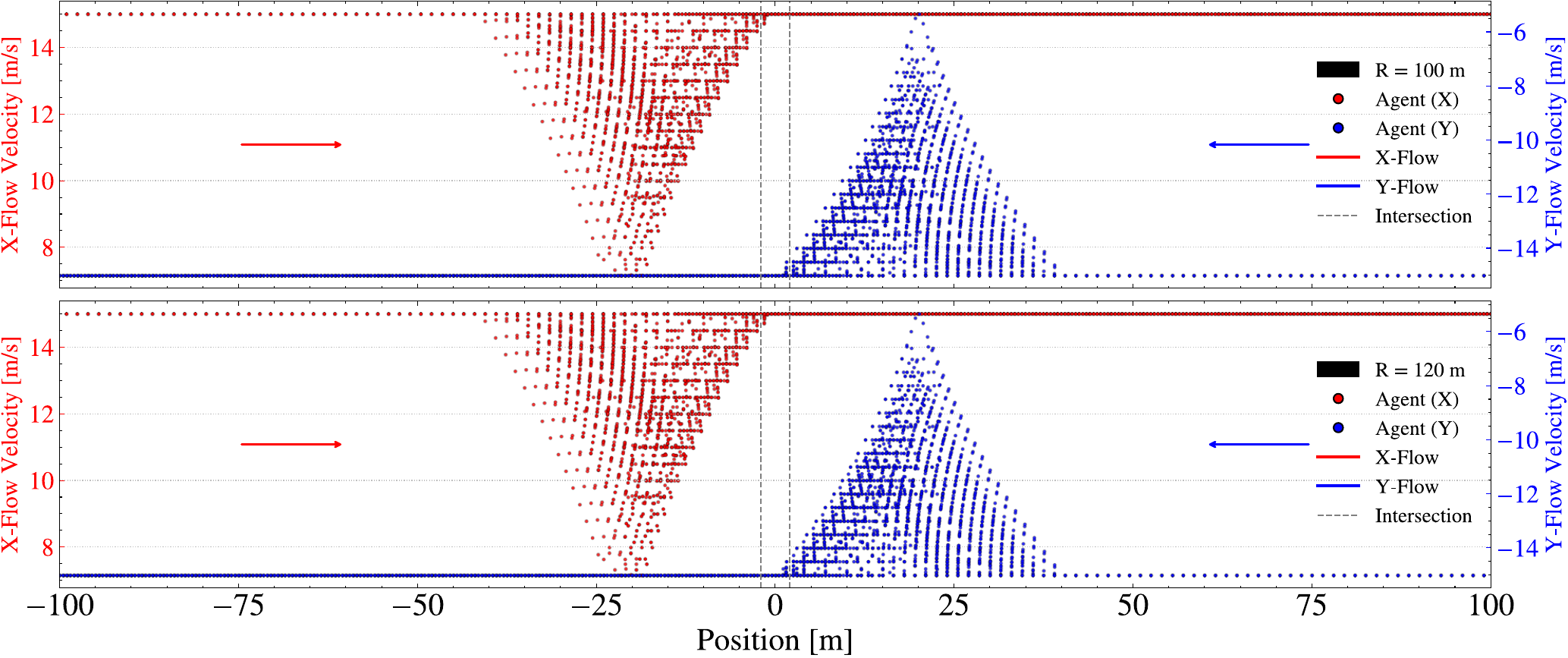}
    \end{overpic} 
    \captionsetup{skip=3pt}
    \caption{Velocity profile for control zone radius of $R=100~\mathrm{m}$ (top plot) and $R=120~\mathrm{m}$ (bottom plot) for traffic density $\Lambda=0.9~\mathrm{s^{-1}}$, for the system without fairness constraint. Red and blue dots denote agents traveling along the x- and y-flows.}
    \label{fig:Velocity_2}
\end{figure*}

\subsection{Platooning Efficiency}
Figs.~\ref{fig:ST_40} and~\ref{fig:ST_80} show the space-time diagrams for the control zones of $R=40~\mathrm{m}$ and $R=80~\mathrm{m}$ for traffic density $\Lambda=0.9~\mathrm{s}^{-1}$. Increasing the control zone gives the controller more time to anticipate conflicts by forming the optimal platoons, reducing the reversal count, and increasing the platoon size.

Table~\ref{tab:R40_A_fourcols} reports $9$ reversals for the configuration highlighted by the black square in Fig.~\ref{fig:ST_40}. The corresponding space–time diagram shows a suboptimal platoon formation near the intersection (grey dots), where two red agents, followed by three blue agents, then three red agents, and finally one blue agent. This fragmented sequencing leads to multiple small platoons and a higher reversal count.  

When the control zone is enlarged to $R=80~\mathrm{m}$, the highlighted region in Fig.~\ref{fig:ST_80} demonstrates a more coherent platoon structure. Here, three red agents, four blue agents, and two red agents are grouped into larger, more organized platoons. The number of reversals decreases to $7$, as reported in Table~\ref{tab:R80_A_fourcols}.

The results show that increasing the control zone radius until it reaches the optimal radius improves the flow, leading to better efficiency by forming optimal platoons with fewer reversal counts.

\begin{figure}[!ht]
    \centering
    \begin{overpic}[width=\columnwidth]{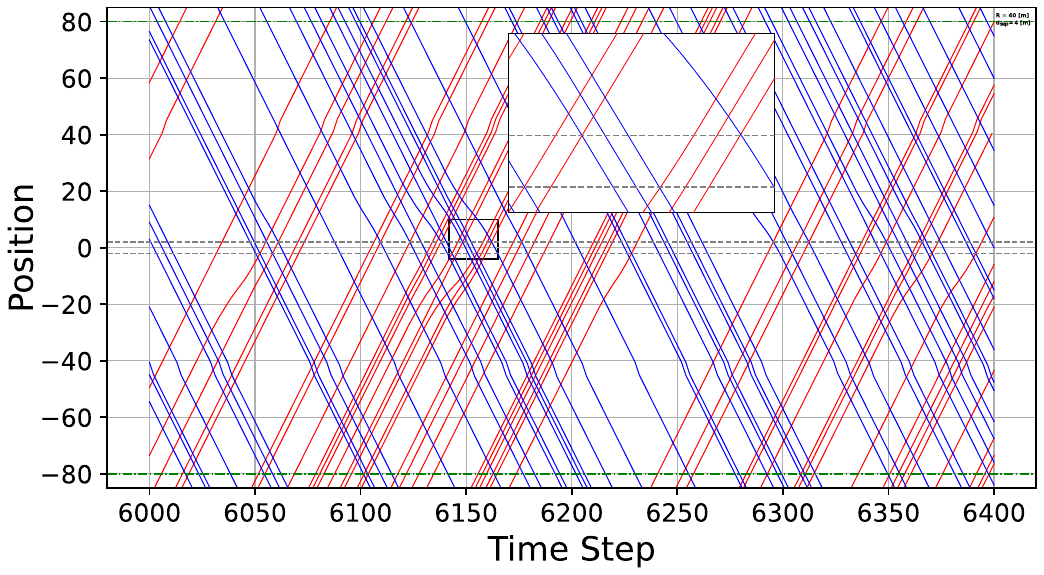}
    \end{overpic} 
    \caption{Excerpt of Space-time diagram of a control zone with a $40~\mathrm{m}$ radius, $d_{\text{safe}}=3~\mathrm{m}$, and $d_{\text{sep}}=4~\mathrm{m}$ and without fairness. Red and blue lines denote agents traveling along the x- and y-flows; the green dashed line indicates the control zone boundary, and the grey area marks the intersection. The highlighted section (indicated by a black square) shows platoon formation. For a dynamic view, please see the \href{https://www.youtube.com/watch?v=mqgiasBgRtE}{video link}.}
    \label{fig:ST_40}
\end{figure}

\begin{figure}[!ht]
    \centering
    \begin{overpic}[width=\columnwidth]{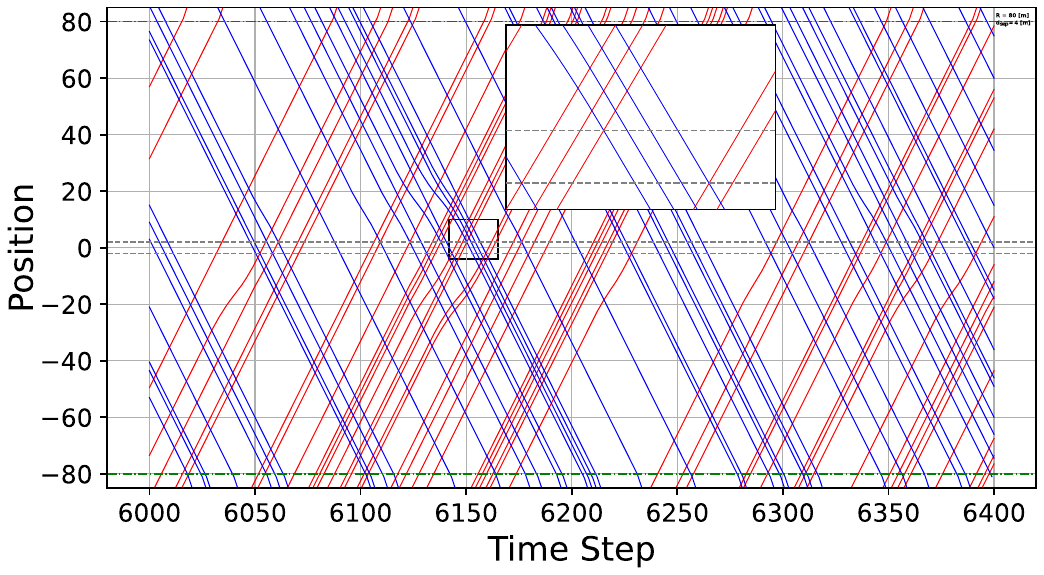}
    \end{overpic} 
    \caption{Excerpt of Space-time diagram of a control zone with a $80~\mathrm{m}$ radius, $d_{\text{safe}}=3~\mathrm{m}$, and $d_{\text{sep}}=4~\mathrm{m}$ and without fairness. Red and blue lines denote agents traveling along the x- and y-flows; the green dashed line indicates the control zone boundary, and the grey area marks the intersection. The highlighted section (indicated by a black square) shows platoon formation. For a dynamic view, please see the \href{https://www.youtube.com/watch?v=mqgiasBgRtE}{video link}.}
    \label{fig:ST_80}
\end{figure}

\begin{table}[ht]
  \centering
  \caption{Arrival vs.\ crossing order at $R = 40 ~\mathrm{m}$, $\Lambda=0.9~\mathrm{s}^{-1}$ (Fig.~\ref{fig:ST_40}).}
  \label{tab:R40_A_fourcols}
  \sffamily
  \begin{tabular}{@{}cccc@{}}
    \toprule
    \textbf{Arrival Order} &
    \textbf{Crossing Order} &
    \textbf{Flow Direction} &
    \textbf{Order shift} \\
    \midrule
    1016 & 1016 & x &  0 \\
    1018 & 1017 & x & -1 \\
    1017 & 1018 & y &  1 \\
    1020 & 1019 & y & -1 \\
    1019 & 1020 & x &  1 \\
    1022 & 1021 & x & -1 \\
    1021 & 1022 & y &  1 \\
    1024 & 1023 & y & -1 \\
    1025 & 1024 & y & -1 \\
    1023 & 1025 & x &  2 \\
    1026 & 1026 & x &  0 \\
    \bottomrule
  \end{tabular}
\end{table}

\begin{table}[ht]
  \centering
  \caption{Arrival vs.\ crossing order at $R = 80 ~\mathrm{m}$, $\Lambda=0.9~\mathrm{s}^{-1}$ (Fig.~\ref{fig:ST_80}).}
  \label{tab:R80_A_fourcols}
  \sffamily
  \begin{tabular}{@{}cccc@{}}
    \toprule
    \textbf{Arrival Order} &
    \textbf{Crossing Order} &
    \textbf{Flow Direction} &
    \textbf{Order shift} \\
    \midrule
    1016 & 1016 & x &  0 \\
    1018 & 1017 & x & -1 \\
    1017 & 1018 & y &  1 \\
    1020 & 1019 & y & -1 \\
    1019 & 1020 & x &  1 \\
    1022 & 1021 & x & -1 \\
    1023 & 1022 & x & -1 \\
    1021 & 1023 & y &  2 \\
    1024 & 1024 & y &  0 \\
    \bottomrule
  \end{tabular}
\end{table}

\section{CONCLUSION} \label{conclusion}

%{\color{red} This paper examines how intervention timing in autonomous cross-flow management affects efficiency and fairness while maintaining operational constraints. We proposed an “early versus late management” approach based on a control zone represented as a circle around the intersection, and used a MILP formulation to optimize system performance.}

%{\color{red} In this paper, we investigated how the timing of centralized intervention shapes the balance between efficiency and fairness in autonomous intersection management. 
%Our approach models the control zone as a circle around the intersection and uses a MILP formulation to jointly optimize safety, efficiency, and fairness objectives.}

{In this paper, we investigated how the timing of centralized intervention shapes the balance between efficiency and fairness in autonomous intersection management. We formulated this coordination problem using an MILP model that optimizes agent trajectories for safety and efficiency within a defined control zone. A central contribution of this work is the integration of a reversal-based fairness constraint within the MILP. This framework enables a quantitative study of how fairness requirements influence overall system performance and platoon formation.}

Our results show that increasing the control zone radius from small values reduces delay and energy consumption up to a moderate range, after which performance saturates and slightly degrades with larger radii. This supports the existence of an optimal radius, which increases with traffic density. In addition, the system with a fairness constraint reduces FIFO order deviations but introduces a measurable cost in delay and energy compared to the system without a fairness constraint, primarily by limiting platoon formation. 
%{\color{red}Fairness active MILP through reversal minimization reduces FIFO order deviations but introduces a measurable cost in delay and energy compared to the system without a fairness constraint, primarily by limiting platoon formation.} These findings highlight a fundamental trade–off between fairness and overall efficiency.

Future work will extend from a single intersection to multiple intersections by moving from two single lanes to four (or more) lanes with asymmetric demands and complex topologies at the intersection. In addition, we will pursue a formal, mathematical proof of the existence of optimal control zone radius and quantify coordination overhead as radius grows, to provide theoretical support for the observed saturation and slight performance degradation at large radii.
\section*{Acknowledgment}
This research was funded by King Abdullah University of Science and Technology’s baseline support (BAS$/1/1682-01-01$).
\bibliographystyle{IEEEtran}

\end{document}